\documentclass{llncs}
\usepackage{amsmath,amssymb,amsfonts,stmaryrd}
\usepackage{graphicx}
\usepackage{xcolor}
\usepackage{hyperref}
\usepackage[T1]{fontenc}
\usepackage{booktabs}
\usepackage{xspace}
\usepackage{listings}
\usepackage{subcaption}
\usepackage{array}
\newif\ifcomments%
\commentsfalse%

\newcommand{\C}{\mathcal{C}}
\newcommand{\G}{\mathcal{G}}

\newcommand{\GC}{\G_\C}
\newcommand{\R}{\mathbb{R}}
\newcommand{\Rp}{\mathbb{R}_{+}}
\newcommand{\N}{\mathbb{N}}

\newcommand{\francois}[1]{\ifcomments\textcolor{blue}{#1}\else #1\fi}
\newcommand{\sylvain}[1]{\ifcomments\textcolor{magenta}{#1}\else #1\fi}
\newcommand{\elisabeth}[1]{\ifcomments\textcolor{cyan}{#1}\else #1\fi}

\DeclareMathOperator*{\argmin}{argmin}

\begin{document}

\title{Graphical Conditions for Rate Independence\\ in Chemical Reaction Networks}

\author{\'{E}lisabeth Degrand \and Fran\c{c}ois Fages \and Sylvain Soliman}

\institute{Inria Saclay-\^Ile de France, Palaiseau, France}

\maketitle
\thispagestyle{plain}

\begin{abstract}
  Chemical Reaction Networks (CRNs) provide a useful abstraction of molecular interaction networks
  in which molecular structures as well as mass conservation principles are abstracted away
  to focus on the main dynamical properties of the network structure.
  In their interpretation by ordinary differential equations,
  we say that a CRN with distinguished input and output species computes a positive real function $f:\Rp\rightarrow\Rp$,
  if for any initial concentration $x$ of the input species,
  the concentration of the output molecular species stabilizes at concentration $f(x)$.
  The Turing-completeness of that notion of chemical analog computation has been established 
  by proving that any computable real function 
  can be computed by a CRN over a finite set of molecular species.
  Rate-independent CRNs form a restricted class of CRNs of high practical value since they enjoy a form of absolute robustness in the sense that
  the result is completely independent of the reaction rates  and depends solely on the input concentrations.
  The functions computed by rate-independent CRNs have been characterized mathematically as the set of piecewise linear functions from input species.
  However, this does not provide a mean to decide whether a given CRN is rate-independent.
  In this paper, we provide graphical conditions on the Petri Net structure of a CRN which entail the rate-independence property
  either for all species or for some output species.
  We show that in the curated part of the Biomodels repository, among the 590 reaction models tested,
  2 reaction graphs were found to satisfy our rate-independence conditions for all species, 94 for some output species, among which 29 for some non-trivial output species.
  Our graphical conditions are based on a non-standard use of the Petri net notions of place-invariants and siphons
  which are computed by constraint programming techniques for efficiency reasons.
\end{abstract}

\section{Introduction}

Chemical Reaction Networks (CRNs) are one fundamental formalism widely used in chemistry, biochemistry,
and more recently computational systems biology and synthetic biology.
CRNs provide an abstraction of molecular interaction networks
in which molecular structures as well as mass conservation principles are abstracted away.
They come with a hierarchy
of dynamic Boolean, discrete, stochastic and differential interpretations~\cite{FS08tcs}
which is at the basis of a rich theory for the analysis of their qualitative dynamical properties~\cite{Feinberg77crt,CF06siamjam,BFS18jtb},
  of their computational power~\cite{CSWB09ab,CDS12nc,FLBP17cmsb},
  and on their relevance as a design method for implementing high-level functions in synthetic biology, using either DNA~\cite{QSW11dna,CDSPCSS13nature}
  or DNA-free enzymatic reactions~\cite{CAFRM18msb,SABEBAFM20aml}.

  In their interpretation by ordinary differential equations,
  we say that a CRN with distinguished input and output species computes a positive real function $f:\Rp\rightarrow\Rp$,
  if for any initial concentration $x$ of the input species,
  the concentration of the output molecular species stabilizes at concentration $f(x)$.
  The Turing-completeness of that notion of chemical analog computation has been shown 
  by proving that any computable real function 
  can be computed by a CRN over a finite set of molecular species~\cite{FLBP17cmsb}.
  
  In the perspective of biochemical implementations with real enzymes however,
  the strong property of rate independence, i.e.~independence of the computed result of the rates of the reactions~\cite{SR11plos},
  is a desirable property that greatly eases their concrete realization, and guarantees a form of absolute robustness of the CRN\@.
  The set of input/output functions computed by a rate-independent CRNs has been characterized mathematically in~\cite{CDS14itcs,CKRS18cmsb} 
  as the set of piecewise linear functions.
  However, this does not give any mean to decide whether a given CRN is rate-independent or not.

In this paper, we provide purely graphical conditions on the CRN structure which entail the rate-independence property either for all molecular species or for some output species.
These conditions can be checked statically on the reaction hypergraph of the CRN, i.e.\ on its Petri net structure,
or can be used as structural constraints in rate-independent CRN design problems.

  \begin{example}\label{ex:min} 
    For instance, the reaction \texttt{a+b=>c} 
    computes at steady state the minimum of $a$ and $b$, i.e.~$c^*=\min(a(0),b(0))+c(0)$, $a^*=\max(0, a(0)-b(0))$, $b^*=\max(0, b(0)-a(0))$
    whatever the reaction rate is.
    Our graphical condition for rate independence on all species
    assumes that there is no synthesis reaction, no fork and no loop in the reaction hypergraph (Thm.~\ref{thm:ri} below).
    This is trivially the case in this CRN and suffices to prove rate-independence for all species in this example.
    \end{example}
 \begin{example}\label{ex:max} 
Similarly, the CRN
\begin{lstlisting}
  a => x+c
  b => y+c
  x+y => z
  c+z => r
\end{lstlisting}
assuming $x(0)=y(0)=c(0)=z(0)=r(0)=0$,
computes at steady state the maximum of $a$ and $b$: $c^*=\max(a(0), b(0))$ (as~$a(0)+b(0)-\min(a(0),b(0))$),
$x^*=\max(0,a(0)-b(0)),\ y^*=\max(0,b(0)-a(0))$,
$z^*=0,\ r^*=\min(a(0), b(0)), \ a^*=0,\ b^*=0$, independently of the reaction rates.
Fig.~\ref{fig:max} shows some trajectories obtained with different values for the mass action law kinetics constants $k_1, k_2, k_3, k_4$ of the four reactions above,
with initial concentrations $a(0)=3, b(0)=1$ and $0$ for the other species.
Here again, our graphical condition is trivially satisfied and demonstrates
the rate-independence property of that CRN for all species, by Thm.~\ref{thm:ri}.
\end{example}

\begin{figure}[htb]
\captionsetup[subfigure]{justification=centering}
 \centering
\begin{subfigure}{0.6\textwidth}
    \centering
    \includegraphics[width=.8\linewidth]{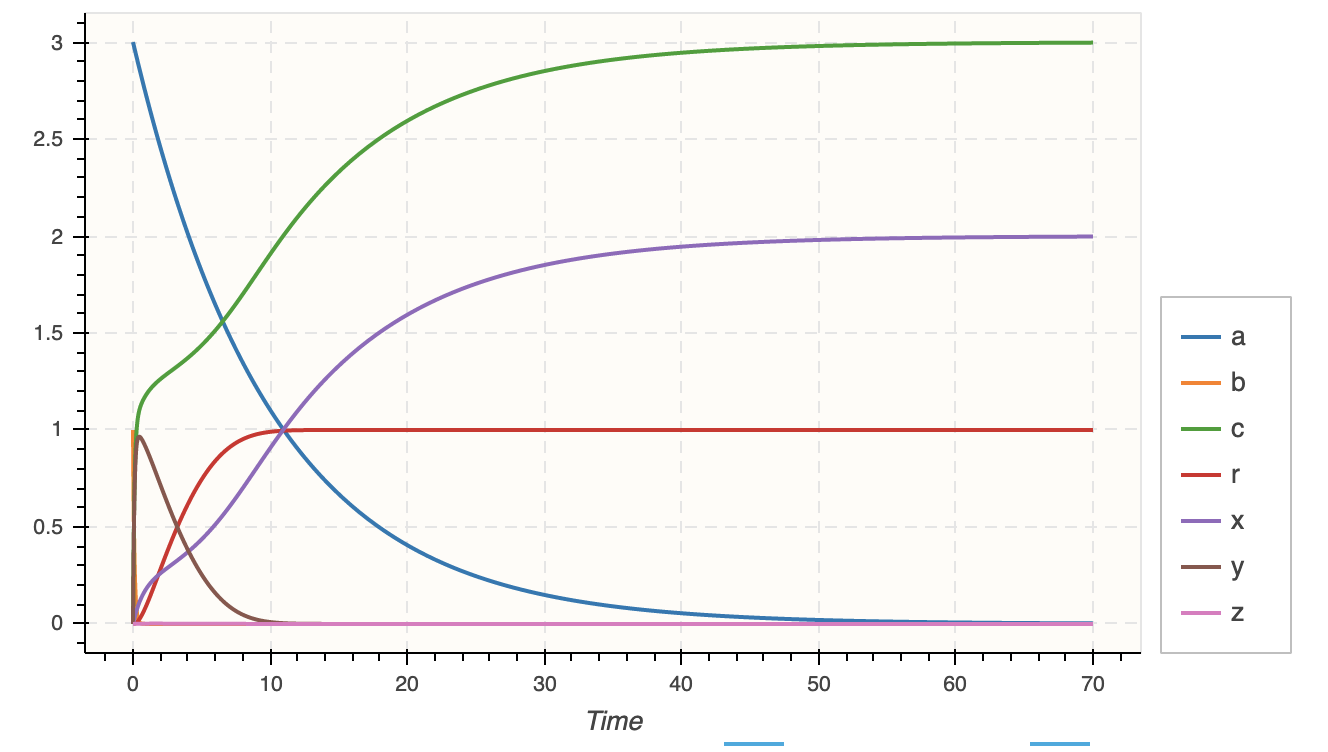}
    \caption{\small\texttt{k1 = 0.1, k2 = 10.0, k3 = 1, k4 = 100.0} }\label{fig:max_1}
\end{subfigure}
\begin{subfigure}{0.6\textwidth}
   \centering
   \includegraphics[width=.8\linewidth]{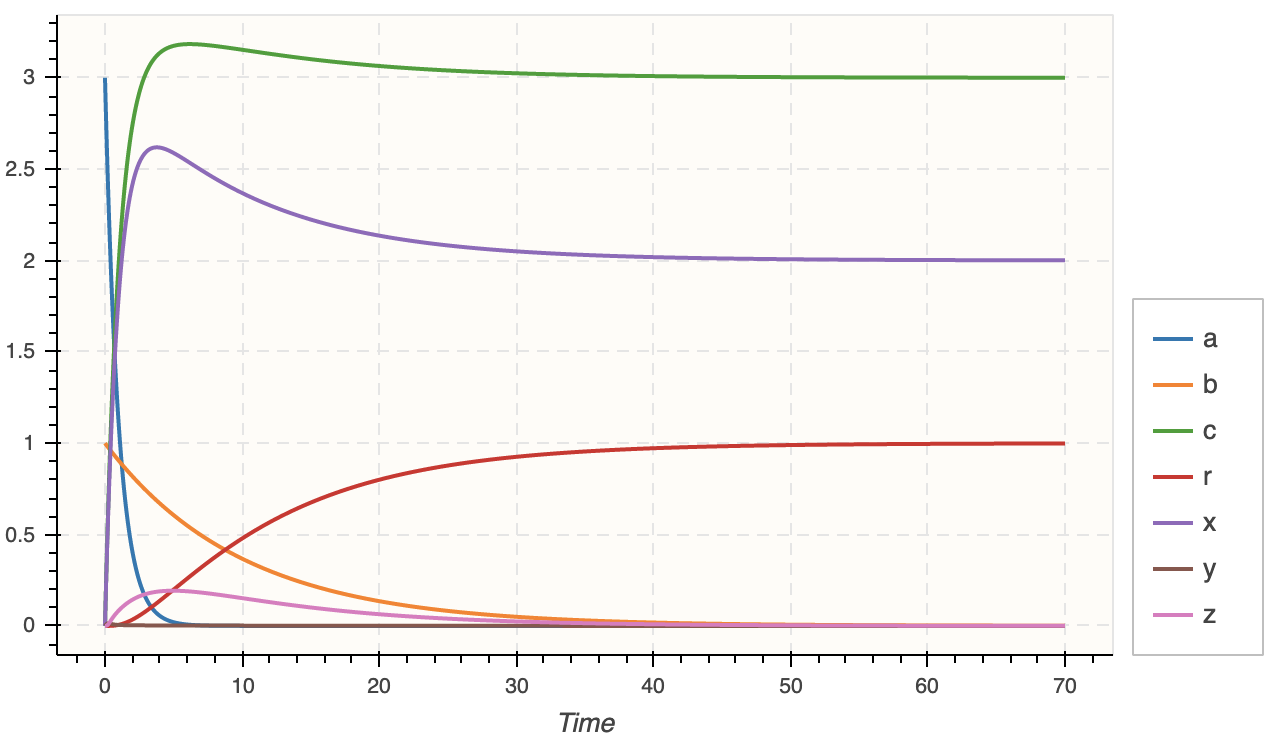}
   \caption{ \small\texttt{ k1 = 1, k2 = 0.1, k3 = 10, k4 = 0.1} }\label{fig:max_2}
 \end{subfigure}
\begin{subfigure}{0.6\textwidth}
   \centering
   \includegraphics[width=.8\linewidth]{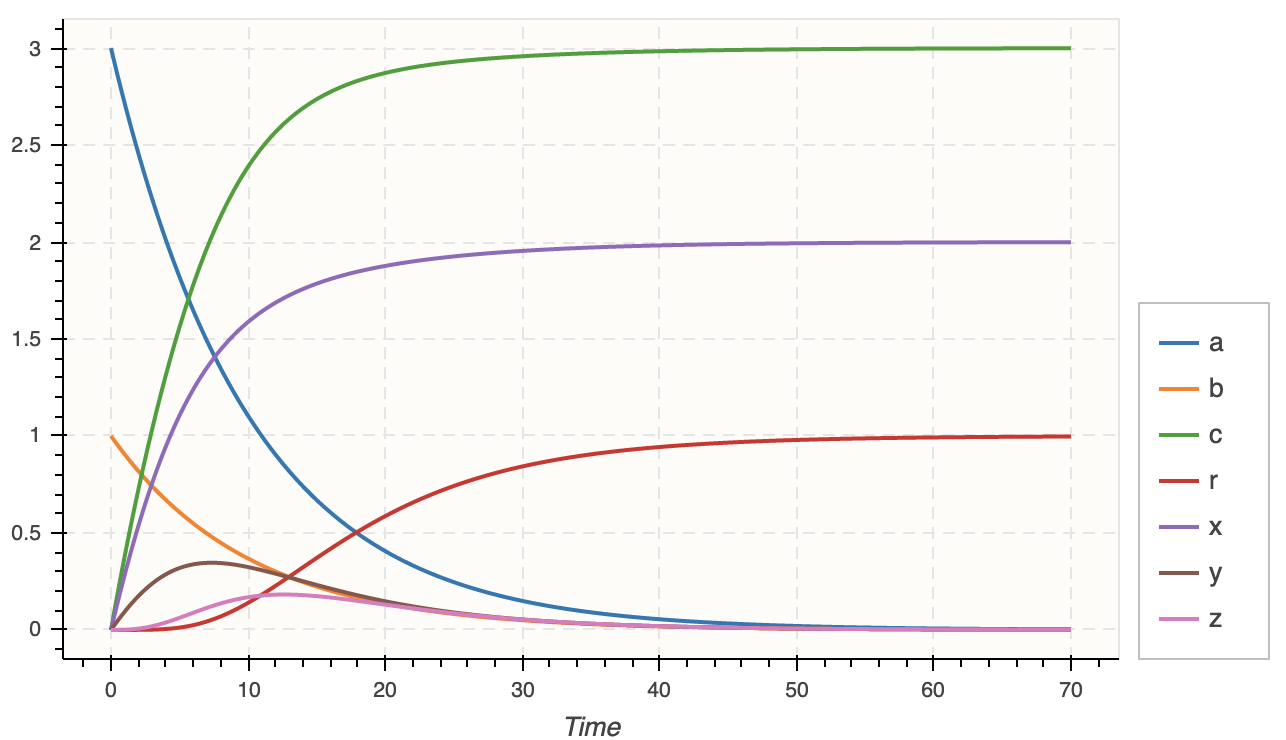}
   \caption{ \small\texttt{ k1 = 0.1, k2 = 0.1, k3 = 0.1, k4 = 0.1} }\label{fig:max_3}
 \end{subfigure}
\caption{Computation of $\max(a,b)$ with the rate-independent CRN of Ex.~\ref{ex:max} with mass action law kinetics
  with different reaction rate constants.}\label{fig:max}
\end{figure}

The rest of this paper is organized as follows.
In Sec.~\ref{sec:siphon}, we first give a sufficient condition for \sylvain{the} rate independence of output species of the CRN.
That condition tests the existence of particular P-invariants and siphons in the Petri net structure of the CRN\@.
\francois{This test is modelled as a constraint satisfaction problem, and implemented using constraint programming techniques}
in order to avoid the enumeration of all P-invariants and siphons that can be in exponential number.
Then in Sec.~\ref{sec:funnel}, we give another sufficient condition that entails the existence of a unique steady state,
and ensures that the computed functions for all species of a CRN are rate-independent.
None of these conditions are necessary conditions but we show with examples that they cover a large class of rate-independent CRNs.
In Sec.~\ref{eval}, we evaluate our conditions on the curated part of the repository of models BioModels~\cite{CLN13issb}
by taking as output species the species that are produced and not consumed.
 We show that 
  2 reaction graphs satisfy our rate-independence conditions for all species, 94 for some output species, among which 29 for some non-trivial output species.
  We conclude on the efficiency of our purely graphical conditions to test rate-independence of existing CRNs,
  and on the possibility to use those conditions as CRN design contraints for synthetic biology constructs such as~\cite{CAFRM18msb}.

\section{Preliminaries} 

\subsection{Notations}

Unless explicitly noted, we will denote sets and multisets by capital letters
(e.g. $S$, also using calligraphic letters for some sets), tuples of
values by vectors (e.g., $\vec x$), and elements of those sets or vectors
(e.g.~real numbers, functions) by small Roman or Greek letters.
\sylvain{For vectors that vary in time, the time will be denoted using a superscript notation like \({\vec x}^{t}\).}
For a multiset (or a set) $M:S\to\N$, $M(x)$ denotes the multiplicity of element
$x$ in $M$ (usually the stoichiometry in the following), and 0 if the element does not belong to the multiset.
By abuse of notation, $\geq$ will denote the integer or Boolean pointwise
order on vectors, multisets and sets (i.e.\ set inclusion), and $+,\ -$ the
corresponding operations for adding or removing elements.
With these unifying notations, set inclusion may thus be noted $S\le S'$ and
set difference $S-S'$. 

\subsection{CRN Syntax}

We recall here definitions from~\cite{FGS15tcs,FS08fmsb} for directed
chemical reactions networks.
In this paper, we assume a finite set $S=\{x_1,\dots,x_n\}$ of molecular species.

\begin{definition}\label{defn:reaction}
A \emph{reaction} over $S$ is a triple $(R, P, f)$, where

\begin{itemize}
   \item $R$ is a multiset of \emph{reactants} in $S$,

   \item $P$ a multiset of \emph{products} in $S$, 

   \item and $f:\R^n\rightarrow\R$ is a \emph{rate function}
      over molecular concentrations or numbers.

\end{itemize}

A \emph{chemical reaction network} (CRN) $\C$ is a finite set of reactions.
\end{definition}

It is worth noting that a molecular species in a reaction can be both a
reactant and a product, i.e.\ a \emph{catalyst}.
Those mathematical definitions are mainly compatible with SBML(\cite{Hucka03bi}), however there
are some differences.
Unlike SBML, we find it useful to consider only directed reactions
(reversible reactions being represented here by two reactions).

Furthermore, we enforce
the following compatibility conditions between the rate function and the
structure of a reaction:

\begin{definition}[\cite{FGS15tcs,FS08fmsb}]\label{defi:rwell}
A reaction $(R, P, f)$ over $S$ is \emph{well-formed} if the following
conditions hold:

\begin{enumerate}

   \item $f$ is a non-negative partially differentiable function,

   \item $x_i\in R$ iff ${\partial {f}}/ {\partial x_i}(\vec x)>0$ for some value
      $\vec x\in\R_+^n$,

   \item $f(x_1,\dots,x_n)=0$  iff \sylvain{there} exists $x_i\in R$ such that  $x_i=0$.

\end{enumerate}
A CRN is well-formed if all its reactions are well-formed.
\end{definition}

Those compatibility conditions are necessary to perform structural analyses of CRN dynamics.
They ensure that the reactants contribute positively to the rate of the reaction at least in some region of the concentration space (condition 2),
that the system remains positive (Prop.~2.8 in~\cite{FGS15tcs})
and that a reaction stops only when one of the reactant has been entirely consumed, whatever the rate function is.

To analyse the notion of function computed by a CRN, we will study the steady states of the ODE system,
i.e., states where $\frac{d\vec{x}}{dt}=0$, the flux $f_i$ of each reaction of
$\C$ at steady state will be called its \emph{steady flux}.

A directed weighted bipartite graph $\GC$ can be naturally associated to a
chemical reaction network $\C$, with species and reactions as vertices,
and stoichiometric coefficients, i.e.~multiplicity in the multisets $R$ and $P$, 
as weights for the incoming/outgoing edges.

\begin{example}\label{LV}
  Fig.~\ref{fig:GC_max} shows the bipartite graph $\GC$ of the Ex.~\ref{ex:max} of the introduction.
For this graph, the weights are all 1 and are not written for that reason.

\begin{figure}
	\centering
	\includegraphics[width=.9\linewidth]{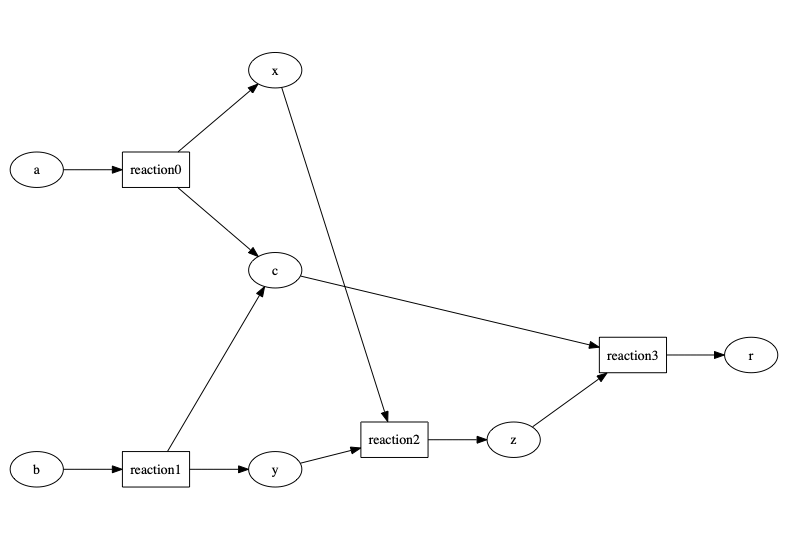}
	\caption{Bipartite graph $\GC$ associated to the CRN given in Example~\ref{ex:max}. {The weights are all equal to 1 and not displayed.}}\label{fig:GC_max}
\end{figure}
\end{example}

\subsection{CRN Semantics}

As detailed in~\cite{FS08tcs}, a CRN can be interpreted in a hierarchy of semantics with different formalisms that can be formally related by abstraction relationships in
the framework of abstract interpretation~\cite{CC77popl}.
In this article, we consider the \emph{differential semantics} which associates with a CRN $\C$ the system ODE($\C$) of Ordinary Differential Equations

\[
   \frac{dx_j}{dt} =
   \sum_{(R_i, P_i, f_i)\in\C}(P_i(j) - R_i(j))\cdot f_i
\]

\begin{example}

Assuming mass action law kinetics for the CRN of Ex.~\ref{ex:max}, the ODEs are:

\begin{align}
   da/dt &= -k_1 \cdot a\\
   db/dt &= -k_2 \cdot b\\
   dc/dt &= k_1 \cdot a + k_2 \cdot b - k_4 \cdot c \cdot z\\
   dr/dt &= k_4 \cdot c \cdot z\\
   dx/dt &= k_1 \cdot a - k_3 \cdot x \cdot y\\
   dy/dt &= k_2 \cdot b - k_3 \cdot x \cdot y\\
   dz/dt &= k_3 \cdot x \cdot y - k_4 \cdot c \cdot z
\end{align}
\end{example}

\begin{definition}~\cite{FLBP17cmsb}
  The function of time computed by a CRN $\C$ from initial state $\vec x\in\R^n$
  is, \sylvain{if} it exists, the solution of the ODE associated to $\C$ with initial conditions $\vec x\in\R^n$
  \end{definition}

\begin{definition}~\cite{FLBP17cmsb}
  The input/output function computed by a CRN \sylvain{with \(n\) species}, on an output species $z$, a set of $m$ input species $\vec y\in\R^m$
  and a fixed initial state for the other species $\vec x\in\R^{n-m}$
  is, \sylvain{if} it exists, the function $f:\R^m\rightarrow \R$
  for which the   ODEs associated to $\C$ have a solution   which moreover stabilizes on some value $f(\vec x,\vec y)$ on the $z$ species component.
  \end{definition}

\begin{definition}
  A CRN is rate-independent on an output species $z$
  if the input/output function computed on $z$ \sylvain{with all species considered as input} does not depend of the rate functions of the reactions.
  \end{definition}

\subsection{Petri Net Structure}

The bipartite graph $\GC$ of a CRN can be naturally seen as a Petri-net graph~\cite{RML93ismb,Chaouiya07bioinfo,CRT08jda},
here used with the continuous Petri-net semantics~\cite{HGD08sfm,GH06icatpn,SHK06bmcbi}.
\sylvain{The species correspond to places and the reactions to transitions of the Petri-net.}
We recall here some classical Petri-net concepts~\cite{Peterson81book,KS06bi}
used in the next section, since they may have various names depending on the community.

\begin{definition}
  A minimal semi-positive \emph{P-invariant} is a vector of $\N^{n}$ that is \sylvain{in the} left-kernel of the stoichiometric matrix.
  Equivalently it is a weighted sum over places concentrations that remains constant by any transition.

  A P-surinvariant is a weighted sum that only increases.

  The \emph{support} of a P-invariant or P-surinvariant is the set of places with non-zero value. Those places will be said to be \emph{covered} by the P-invariant or P-surinvariant.

\end{definition}

Intuitively a P-invariant is a conservation law of the CRN\@.
The notion of P-surinvariant will be used to identify the output species of a CRN.

\begin{definition}
  A \emph{siphon} is a set of places such that for each edge from a transition to any place of the siphon, there is an edge from a place of the siphon to that transition.
\end{definition}

Intuitively a siphon is a set of places that once empty remains empty\sylvain{, i.e., a set of species that cannot be produced again once they have been completely consumed}.
Our first condition for rate independence will be based on the following.

\begin{definition}
  A \emph{critical siphon} is a siphon that does not contain the support of any P-invariant.
\end{definition}

\sylvain{A siphon that is not critical contains the support of a P-invariant, therefore it cannot ever get empty. A critical siphon on the other hand is thus a set of species that might disappear completely and then always remain absent.}

\section{Rate Independence Condition for Persistent Outputs}\label{sec:siphon}

The persistence concept has been introduced to identify Petri nets for which places remain non-zero~\cite{ALS07bct}.
Here we etablish a link between this notion of persistence and the rate-independence property of the input/output function computed on some output species.

\subsection{Sufficient Graphical Condition} 

As in~\cite{ALS09cdc}, we are interested by the persistence not of the whole CRN but of some
species.
\francois{We will say that a species is an output of a CRN if it is  produced and not consumed (and thus can only increase),
  i.e.~if the stoichiometry of that species in the product part of any reaction
  is greater or equal to the reactant part}, or equivalently:

\sylvain{
  \begin{definition}\label{def:output}
    A species is an \emph{output} of a CRN if it is the singleton support of a P-surinvariant.
  \end{definition}
}

\begin{example}
  In the CRN of Ex.~\ref{ex:max}, the \(max\) is computed on a non-output node, \(c\).
  The \(min\) is computed on an output node \(r\).
  The rate independence of that CRN on \(r\) follows from Thm.~\ref{thm:poire} below.
  \end{example}
\begin{definition}
  A species is \emph{structurally persistent} if it is covered by a P-invariant and does not belong to any critical siphon.
\end{definition}

\sylvain{Such species' concentrations will not reach zero for well-formed CRNs as proved in~\cite{ALS07bct,ALS09cdc}, but this section shows that if they are also output species they converge to a value that is independent of the rates of reactions.}
\sylvain{Note that such species might still belong to some non-critical siphons, for instance siphons that cover the whole P-invariant it is part of.}

\begin{theorem}\label{thm:poire}

  If a species \(p\) of a well-formed CRN is a structurally persistent output, then that CRN is rate-independent on \(p\).

\end{theorem}

\begin{proof}

  Since \(p\) is structurally persistent, it is covered by some P-invariant and therefore bounded.
  Since \(p\) is an \sylvain{\emph{output species}} and the CRN is well-formed, \(\frac{dp}{dt}\geq 0\).
  Hence its concentration converges to some value \(p^{*}\).

  When \(p\) reaches that steady state, all incoming reactions \sylvain{that modify it} have null flux, hence by well-formed-ness
  one of their reactants has \(0\) concentration. \sylvain{If there are only incoming reactions that do not affect \(p\) then it is trivially constant and therefore rate-independent. Otherwise there are some such incoming reactions with a null reactant.}

  Now, \sylvain{notice that} Prop.~\sylvain{1} of~\cite{ALS07bct} \sylvain{states, albeit with completely different notations,
  that if one species of a well-formed CRN reaches 0 then all the species of a siphon reach 0. Therefore}
  there exists a whole siphon \(S\) containing that reactant
  and with \(0\) concentration (intuitively, this reactant also has its input fluxes null,
  and one can thus build recursively a whole siphon).

  By construction, \(S' = \{p\}\cup S\) is also a siphon, and since \(p\) is persistent, \(S'\)
  is not critical. \(S'\) therefore covers some P-invariant \(P\) and all concentrations are null
  except that of \(p\) in \(S'\). Now \(P\) necessarily covers \(p\) since otherwise its conservation would
  be violated by having all \(0\) concentrations.

  Note that by definition, for each P-invariant \(V\) containing \(p\) we have at any time \(t\)
  \sylvain{with state vector \({\vec x}^{t}\) that \(V\cdot {\vec x}^{t} = V\cdot {\vec x}^{0}\). Hence:}

  \[p(t) = {\vec x}^{t}(p) = \frac{V^{0}-V\cdot {\vec x}_{p\rightarrow 0}^{t}}{V(p)}\leq\frac{V^{0}}{V(p)}\]
  \sylvain{where \({\vec x}_{p\rightarrow 0}^{t}\) is the state vector except for the concentration of \(p\) replaced by \(0\), and
  \(V^{0}\) is a shorthand for \(V\cdot {\vec x}^{0}\).}

  \sylvain{At steady state we get \(p^{*}=\frac{P^{0}}{P(p)}\) since we proved that all concentrations other than that of \(p\)
  are null. Hence, we have:}

  \[p^{*}=\min_{W}\frac{V^{0}}{V(p)}\]
  \sylvain{where \(W = \{V\mid V\text{ is a P-invariant covering }p\}\), and which is obviously rate-independent.\hfill\qed}%

\end{proof}

\subsection{Constraint-based Programming}\label{sec:cp}

It is well-known that there may be an exponential number of P-invariants and siphons in a Petri net.
Therefore, it is important to combine the constraints of both structural conditions for the
computation of the minimal P-invariants and the union of critical siphons,
without computing all siphons and P-invariants.
This is the essence of constraint programming and of
constraint-based modeling of such a decision problem as a constraint satisfaction problem.
Furthermore, deciding the existence of a minimal siphon containing a given place
is an NP-complete problem for which constraint programming has already shown its practical efficiency
for enumerating all minimal siphons in BioModels, see~\cite{NMFS16constraints}.

We have thus developed a \emph{constraint program} dedicated to the computation of structurally persistent species.
For the minimal P-invariants, the constraint solving problem
is the same as in~\cite{Soliman12amb} and is quite efficient on CRNs.
For the second part about critical siphons, we use a similar
approach but with Boolean variables to represent our siphons as in~\cite{NMFS16constraints}.
However, we enumerate \emph{maximal} siphons here.
This amounts to enumerate values \(1\) before \(0\), and to add in the
branch-and-bound procedure for optimization that each new siphon must include at least one new place.
Furthermore, we add the constraint
that they are critical: for each P-invariant \(P\), one of the species of its support must be absent (\(0\)).
We get the flexibility of our constraint-based approach to add this kind of supplementary constraint while
keeping some of the efficiency already demonstrated before.

In Section~\ref{eval}, this constraint program is used to compute the set of outputs and check if they are structurally persistent for many models
of the \texttt{biomodels.net} repository. 
There are however a few models on which our constraint program is quite slow.
An alternative constraint solving technique to solve those hard instances
could be to use a SAT solver, at least for the enumeration of critical siphons, as shown in~\cite{NMFS16constraints}.

\section{Global Rate Independence Condition}\label{sec:funnel}

Thm~\ref{thm:poire} above can be used to prove the rate-independence property on some output species of a CRN,
like \(r\) in Ex.~\ref{ex:max} for computing the \(max\),
but not on some intermediate species, like \(c\) for computing \(min\).
In this section we provide a sufficient condition for proving the rate-independence of a CRN on all species\@.

\subsection{Sufficient Graphical Condition}

\begin{definition}
   A chemical reaction network $\C$ is \emph{synthesis-free} if for all reactions
   $(R_i, P_i, f_i)$ of $\C$ we have $R_i\not\leq P_i$.
\end{definition}

\sylvain{In other words any reaction need to consume something to produce something.}

\begin{definition}

   A chemical reaction network $\C$ is \emph{loop-free} if there is no circuit
   in its associated graph $\GC$.

\end{definition}

\begin{definition}

  A chemical reaction network $\C$ is \emph{fork-free} if for all species $x\in S$
  there is at most one reaction $(R_i, P_i, f_i)$ such that $R_i(x)>0$.

\end{definition}

This is equivalent to saying that the out-degree of species vertices is at
most one in $\GC$.

\begin{definition}\label{def:brocoli}
   A \emph{funnel} CRN is a CRN that is: 
   \begin{enumerate}
      \item synthesis-free
      \item loop-free
      \item fork-free
   \end{enumerate}
\end{definition}

In Ex.~\ref{ex:max} for computing the maximum concentration of two input species, $A$ and $B$,
one can easily check that the CRN statisfies the funnel condition (see Fig.~\ref{fig:GC_max}).
More generally, we can prove that any well-formed funnel CRN has a single stable state
and that this state does not depend on the precise values of the parameters of
the rate functions $f_i$.

\begin{lemma}\label{lemma:dag}

   The structure of the bipartite graph $\GC$ of a funnel CRN $\C$ is a
   DAG with leaves that are only species.

\end{lemma}

\begin{proof}
   Since $\C$ is loop-free, $\GC$ is acyclic.
   Since $\C$ is synthesis-free, leaves cannot be reactions.\hfill\qed%
\end{proof}

\begin{lemma}\label{lemma:fluxes}
   All steady fluxes of a funnel CRN $\C$ are equal to $0$.
\end{lemma}

\begin{proof}

   Let us prove the lemma by induction on the topological order of reactions
   in $\GC$, this is enough thanks to Lemma~\ref{lemma:dag}.

   For the base case (smallest reaction in the order), at least one of the
   species $x$ such that $R_i(x)>P_i(x)$ is a leaf (synthesis-freeness),
   then notice that at steady state,
   $\frac{dx}{dt} = 0 = (P_i(x)-R_i(x))f_i$
   since there is no production of $x$ as it is a leaf, and no other
   consumption as $\C$ is fork-free.
   Hence $f_i = 0$.

   For the induction case, consider a reactant
   $x$ s.t.\ $R_i(x)>P_i(x)$ of our reaction.
   By induction hypothesis, at steady state we have
   $\frac{dx}{dt} = \sylvain{0 =} (P_i(x)-R_i(x))f_i$
   since all productions of $x$ are \sylvain{lower in the topological order}, and
   there is no other consumption of $x$ as $\C$ is fork-free.
   Hence $f_i=0$.\hfill\qed%

\end{proof}

\begin{definition}\label{def:xplus}
   
   We shall denote $x_i^+$ the total amount of species $x_i$ available in an
   execution of the corresponding ODE system.

   \[
      x_i^+ = x_i^0 + \int_0^{+\infty} \frac{dx_i^+}{dt} = x_i^0 +
      \int_0^{+\infty}\sum_{P_j(x_i)>R_j(x_i)}(P_j(x_i)-R_j(x_i))f_j
   \]

\end{definition}

\begin{lemma}\label{lemma:xplus}

   Let $\C$ be a well-formed funnel CRN, then for each initial state
   $\vec x^0$, if $\C$ reaches a steady state $\vec x^*$, then the total
   amount $x_i^+$ of any species $x_i$ can be computed and is independent
   from the kinetic functions $f_j$ of $\C$.

\end{lemma}

\begin{proof}
   
   Let us proceed by induction on the topological order of species $x_i$ in
   $\GC$.

   If $x_i$ is a leaf, then since nothing produces it $x_i^+ = x_i^0$.

   Now let us look at the induction case for $x_i$, and consider
   the set $J$ of reactions producing $x_i$ (i.e., such that
   $P_j(x_i)>R_j(x_i)$).

   From Lemma~\ref{lemma:fluxes} we know that for all these reactions $f_j =
   0$ at stable state, and since $\C$ is well-formed, it means that there
   exists at least one species $x_{j_0}$ such that $x_{j_0}^*=0$.
   As $\C$ is fork-free and well-formed $x_{j_0}$ has only been consumed by
   reaction $r_j$, which
   led to precisely producing an amount of $x_i$ equal to
   $x_{j_0}^+(P_j(x_i)-R_j(x_i))/(R_j(x_{j_0})-P_j(x_{j_0}))$,
   where $x_{j_0}^+$ is available via induction hypothesis.
   Note also that $j_0 = \argmin_{x_k\mid R_j(x_k)>P_j(x_k)} x_k^+(R_j(x_k)-P_j(x_k))$
   since the reaction will stop as soon as it has depleted one of its inputs.

   Hence $x_i^+ = x_i^0+\sum_J
   x_{j_0}^+(P_j(x_i)-R_j(x_i))/(R_j(x_{j_0})-P_j(x_{j_0}))$,
   which only depends on the initial state and the stoichiometry.\hfill\qed%

\end{proof}

\begin{theorem}\label{thm:ri}

   Let $\C$ be a well-formed funnel CRN, then the ODE system associated to
   $\C$ has a single steady state $\vec x^*$ that
   does not depend on the kinetic functions $f_i$ of $\C$.

\end{theorem}

\begin{proof}

   From proof of Lemma~\ref{lemma:xplus} one notices that
   either $x_i$ is not consumed at all and we have $x_i^*=x_i^+$ or if $j$ is the only
   reaction consuming $x_i$, its total consumption is given by
   $x_{j_0}(R_j(x_{j_0})-P_j(x_{j_0}))$, with $j_{0}$ defined as in the proof of Lemma~\ref{lemma:xplus}.

   These $x_i^*$ do not depend on the kinetic functions $f_i$ of $\C$.

  We prove now that every $x_i$ is convergent. It can be first noticed that
  \[x_i(t) = x_{i_0} + F_e(t) - F_s(t)\]
  where $F_e(t) = \int_0^{t}\sum_{P_j(x_i)>R_j(x_i)}(P_j(x_i)-R_j(x_i))f_j$ is the incoming flux
  and $F_s(t) = \int_0^{t}(P_k(x_i)-R_k(x_i))f_k$ is the outgoing flux.\\
  $F_e$ is the integral of a positive quantity, it is then an increasing function.
  Moreover, as $F_e(t) \le x_i^+$, this function is bounded and then converges to a real number limit.\\
  Similarly, $F_s$ is increasing and, as $x_i(t)\ge 0$, we have $F_s(t)\le x_i^+$ then it is bounded and converges.\\
  To conclude, $x_i$ is a difference of two convergent functions, hence it converges to a real number.
\end{proof}
\begin{corollary}
  Any well-formed funnel
  CRN is rate-independent for any output species.
\end{corollary}

We have thus given here a sufficient condition for a very strong notion of \emph{rate-independence}
in which all the species of the CRN have a steady state independent of the reaction rates, as in Ex.~\ref{ex:max}.

\subsection{Necessary Condition}

Our sufficient condition is not a necessary condition for global rate independence.
Basically, forks that join and circuits that leak do not prevent rate independence:

\begin{example}\label{ex:leak_loop}
The CRN
\begin{lstlisting}
a=>b.
b=>a.
b=>c.
\end{lstlisting}
is not a funnel CRN as it has both a loop (formed by $a$ and $b$) and a fork ($b$ is a reactant in two distinct reactions).
Nevertheless, this CRN is rate-independent on all species. The circuit formed by $a$ and $b$ has a leak with the third reaction.
Every molecule of $a$ and $b$ will thus be finally transformed into $c$ whatever the reaction kinetics are.
At the steady state, the concentration of $a$ and $b$ will be null, and the concentration of $c$ will be the sum of all the initial concentrations.
\end{example}

Nevertheless, we can show that any function computable by a rate independent CRN can be computed by a funnel CRN\@.
We first show that funnel CRNs are composable under certain conditions for rate independent CRNs,
similarly to the composability conditions given in~\cite{CKRS18cmsb}.

\begin{definition}
  Two CRNs $\C_{1}$ and $\C_{2}$ are \emph{composable} if
  \[(\bigcup_{(R, P, f)\in\C_{1}}R\cup P)\cap (\bigcup_{(R',P', f')\in\C_{2}}R'\cup P') = \{x\}\]
  i.e., there is a single species appearing in both sets of reactions.

  The \emph{composition} of $\C_{1}$ and $\C_{2}$ is the union of their sets of reactions.
  The species $x$ is called the \emph{link} between both CRNs
\end{definition}

\begin{lemma}\label{lemma:compo}
The composition of two funnel CRNs is a funnel CRN if the composition does not create forks on their link.
\end{lemma}
\begin{proof}
As the reaction rates of the two original CRNs are well-formed, the reaction rates of the resultant CRN are well-formed too.
No synthesis and no loop can be created by the union of two CRNs as all species are different except for the link.
Therefore, the condition to create no fork by composition is sufficient to ensure that the resultant CRN is a funnel CRN\@.
\end{proof}

\begin{corollary}\label{cor:compo}
The composition of two funnel CRNs is a funnel CRN if the link $x$ is reactant in at most one of the CRNs.
\end{corollary}
\begin{proof}
Since both CRNs are funnel, $x$ is a reactant in at most one reaction in each.
Now from our hypothesis it is not reactant at all in one of the CRNs, hence it appears as reactant in at most
one reaction and therefore in no fork.
By Lemma~\ref{lemma:compo}, the resulting CRN is a funnel CRN\@.
\end{proof}

\begin{theorem}
Any function computable by a rate independent CRN is computable by a funnel CRN\@.
\end{theorem}
\begin{proof}
Using the same theorem from Ovchinnikov~\cite{Ovchinnikov02cag} as
in~\cite{CDS14itcs} we note that any such function $f$ with components
$f_j, 1\leq j\leq p$ can be written as
$f(x) = \max_{1\leq i\leq q}\min_{j\in S_{i}}f_{j}(x)$ for some family
$S_{i}\leq \{1,\ldots,p\}$.

Each $f_j$ is rational linear, so this function can be written: $f_j(x)=\sum_1^{n} \frac{\alpha_{j,i}}{n_j}x_i$.
To compute this linear sum, the following reactions are needed:
For every $x_i$, we add the reactions
$x_i =>\alpha_{j,i} \cdot w_j$
which compute $w=\sum_1^{n} \frac{\alpha_{j,i}}x_i$.

Then we add the reaction
$n_j \cdot w_j => y_j$
which compute $y_j = \frac{1}{n_j} w_j$.

The output of the CRN that computes a linear function is a funnel CRN\@.
Both $\max$ and $\min$ can be written with a
funnel CRN (see respectively Example~\ref{ex:max} and Example~\ref{ex:max})
and $\min$ can be composed by $\max$ as the output of $\min$ is not a reactant in the CRN that computes $\min$.
From Corollary~\ref{cor:compo}, the conclusion is immediate.
\end{proof}

\section{Evaluation on Biomodels}\label{eval}

In this section, we evaluate our sufficient condition for rate-independence
on the reaction graphs of the curated part of the repository of models BioModels~\cite{CLN13issb}.
These models are numbered from \verb=BIOMD0000000001= to \verb=BIOMD0000000705=.
After excluding the empty models (i.e.~models with no reactions or species), 590 models have been tested in total.
As already noted in~\cite{FGS15tcs} however,
many models in the curated of BioModels come from ODE models that have not been transcribed in SBML with well-formed reactions.
Basically, some species appearing in the kinetics are missing as reactants or modifiers in the reactions,
or some kinetics are negative.
In this section, we test our graphical conditions for rate independence on the reaction graphs given for those models,
without rewriting the structure of the reactions when they were not well-formed.
Therefore, the actual rate-independence of the models that satisfy our sufficient criteria is conditioned to the well-formedness of the CRN.

The evaluation has been performed using Biocham\footnote{All our experiments are available on \url{https://lifeware.inria.fr/wiki/Main/Software\#CMSB20b}}
with a timeout of 240 seconds.
The computer used for the evaluation has a quad-processor Intel Xeon 3.07GHz with 8Gb of RAM\@.

\subsection{Computation of rate-independent output species}

Following Def.~\ref{def:output}, we tested the species that constitute the singleton support of a P-surinvariant.
Among the 590 models tested,  340, i.e.~57.6\% of them, were found to have no output species.
94 models, i.e.~15.9\% of the models, were found to have at least one rate-independent output.
27 models, i.e.~4.5\%,  have both one rate-independent outut and one undecided output, i.e. an output not satisfying our sufficient condition.
86 models, i.e.~14.5\%, have at least one undecided output.

It is worth noting however that the species that are never modified by a reaction, i.e.~that are only catalysts, remain always constant and thus constitute trivial rate-independent outputs.
Amongst the 94 models with at least one rate-independent output found during evaluation, 29 have at least one non-trivial rate-independent output.
Table~\ref{table:poire} gives some details on the size and computation time for those 29 models.
\begin{table}[h!]
\begin{center}
\begin{tabular}{|r|r|r|r|r|r|c|r|}
 \hline
Biomodel\# & \#species & \#reactions & \#outputs & \#RI & \#NTRI & NTRI-species & Time (s) \\ 
\hline
037 & 12 & 12 & 2 & 2 & 2 & Yi, Pi & 0.950\\
\hline
104 & 6 & 2 & 3 & 3 & 1 & species\_4 & 0.074\\
\hline
105 & 39 & 94 & 11 & \elisabeth{3} & 1 & AggP\_Proteasome & 63.366\\
\hline
143 & 20 & 20 & 4 & \elisabeth{1} & 1 & MLTH\_c & 3.333\\
\hline
178 & 6 & 4 & 1 & 1 & 1 &  lytic & 0.139\\
\hline
227 & 60 & 57 & 2 & \elisabeth{1} & 1 &  s194 & 17.299\\
\hline
259 & 17 & 29 & 1 & 1 & 1 &  s10 & 2.308\\
\hline
260 & 17 & 29 & 1 & 1 & 1 &  s10 & 2.310\\
\hline
261 & 17 & 29 & 1 & 1 & 1 &  s10 & 2.297\\
\hline
267 & 4 & 3 & 1 & 1 & 1 &  lytic & 0.086\\
\hline
283 & 4 & 3 & 1 & 1 & 1 &  Q & 0.053\\
\hline
293 & 136 & 316 & 14 & \elisabeth{4} & 3 & aggE3, aggParkin,&\\
&&&&&&AggP\_Proteasome & >240\\
\hline
313 & 16 & 16 & 4 & \elisabeth{2} & 1 & IL13\_DecoyR & 2.071\\
\hline
336 & 18 & 26 & 1 & 1 & 1 & IIa & 4.148\\
\hline
344 & 54 & 80 & 7 & \elisabeth{2} & 1 & AggP\_Proteasome & >240\\
\hline
357 & 9 & 12 & 1 & 1 & 1 & T & 0.561\\
\hline
358 & 12 & 9 & 4 & \elisabeth{2} & 1 & Xa\_ATIII & 0.892\\
\hline
363 & 4 & 4 & 1 & 1 & 1 & IIa & 0.067\\
\hline
366 & 12 & 9 & 4 & \elisabeth{2} & 1 & Xa\_ATIII & 0.901\\
\hline
415 & 10 & 5 & 7 & 7 & 7 & s10, s11, s12,&\\
&&&&&&s13, s14, s9, s15 & 0.894\\
\hline
437 & 61 & 40 & 22 & \elisabeth{8} & 1 & T & 16.109\\
\hline
464 & 14 & 10 & 6 & \elisabeth{3} & 1 & s12 & 2.282\\
\hline
465 & 16 & 14 & 5 & 5 & 1 & s23 & 59.554\\
\hline
525 & 18 & 19 & 8 & \elisabeth{3} & 1 & p18inactive & 33.479\\
\hline
526 & 18 & 19 & 8 & \elisabeth {3} & 1 & p18inactive & 33.858\\
\hline
540 & 22 & 11 & 12 & \elisabeth{11} & 8 & s14, s15, s16, s17, &\\
&&&&&&s18, s19, s20, s21 & 56.134\\
\hline
541 & 37 & 32 & 13 & 9 & 7 & s14, s15, s16, s17, & \\
&&&&&&s18, s19, s21 & 31.573\\
\hline
559 & 90 & 136 & 18 & \elisabeth{2} & 2 & s493, s502 & 150.954\\
\hline
575 & 76 & 58 & 9 & \elisabeth{1} & 1 & DA\_GSH & 66.806\\
\hline

\end{tabular}
\end{center}
\caption{
  Model numbers in Biomodels containing non-trivial structurally persistent output species which are thus rate-independent by Thm.~\ref{thm:poire}.
For each model, we indicate the numbers of species, reactions, rate-independent species, non-trivial rate-independent species and total computation time in seconds.}\label{table:poire}
\end{table}

Now, evaluating by simulation the actual rate-independence property of those models,
  and thereby the empirical completeness of our purely graphical criterion in this benchmark, would raise a number of difficulties.
  First, as said above, many SBML models coming from ODE models have not been properly transcribed with well-formed reactions and would need to be rewritten \cite{FGS15tcs}.
Second, some models may contain additional events or assignment rules which are not reflected in the CRN reaction graph.
Third, the relevant time horizon to consider for simulation is not specified in the SBML file.
In the curated part of BioModels, this time horizon can range from 20s to 1 000 000s.

Nevertheless, we performed some manual testing on 9 models from Table~\ref{table:poire}, namely models 37, 104, 105, 143, 178 and 227,
which have at least one non-trivial rate-independent output, and models 50, 52 and 54, which have only undecided outputs.
For each model, numerical simulations were done with two different sets of initial concentrations and two different sets of parameters.
Even when it was not the case in the original models, all the parameters \sylvain{were set to positive values}.
All outputs in models 37 and 104 were found rate-independent which was confirmed by numerical simulation.
For model 105, 3 outputs among the 11 outputs of this model were found rate-independent by our algorithm which seemed again to be confirmed by numerical simulation.
Models 143 and 227 are not well-formed which explains why the species satisfying our graphical criterion were shown not tbe rate-independent by numerical simulation.
For models with only undecided outputs, i.e.~models 50, 52 and 54, numerical simulations show that none of their outputs is rate-independent. For these 3 models, 11 undecided outputs were tested in total.
In this manual testing, we did not find any output that was left undecided by the algorithm and was found rate-independent by numerical simulation.

\subsection{Test of global rate-independence}
In this section, we test the criterion given in Def.~\ref{def:brocoli} that ensures the rate-independence of all the species of a given CRN\@. 

On the 590 reaction models tested, 20 models have reached the timeout limit of 240 seconds and were therefore not evaluated.
Two models were found to be rate-independent on all species, namely models \verb=BIOMD0000000178= and \verb=BIOMD0000000267=.
These models constitute a chain of respectively 4 and 3 species.
At steady state, all species have a null concentration, except the last one.
The steady state value of the last species is equal to the sum of all the initial concentrations.
These models simulate the onset of paralysis of skeletal muscles induced by botulinum neurotoxin serotype A.
They are used in particular to get an upper time limit for inhibitors to have an effect~\cite{LAEC08jpp}.

These two models were also found to have rate-independent outputs during the evaluation of the previous criterion for outputs. 
The global criterion here shows that not only the output species of the chain are rate-independent, but also all the inner species of the chain.

\section{Conclusion}

We have given two graphical conditions for verifying the rate-independence property of a chemical reaction network.
First, the absence of synthesis, circuit and fork in the reaction graph,
ensures the existence of a single steady state that does not depend on the reaction rates,
thereby ensuring the existence of of a computed input/ouput function for all species of the CRN and thier independence of the rate of the reactions.
Second, the covering of a given output species by one P-invariant and no critical siphon,
provides a criterion to ensure the rate-independence property of the computed function on that output species.

These graphical conditions are sufficient but none of them is necessary.
Evaluation in BioModels suggests however that they are already quite powerful
since among the 590 models of the curated part of BioModels tested,
94 reaction graphs were found rate-independent for some output species, 29 for non-trivial output species, and 2 for all species
which was confirmed for well-formed models.

It is worth noting that our second condition uses the classical Petri net notions of P-invariant and siphons 
in a non-standard way for continuous systems.
A similar use has already been done for instance in~\cite{ALS07bct} for the study of persistence and monotone systems,
and interestingly in~\cite{JACB18jmb}, where the authors remark the discrepancy there is on the Petri net property of trap
between the standard discrete interpretation, under which a non empty trap remains non empty,
and the continuous interpretation under which a non empty trap may become empty.
This shows the remarkable power of Petri net notions and tools for the study of continuous dynamical systems,
thus beyond standard discrete Petri nets and outside Petri net theory properly speaking.

As already remarked in previous work~\cite{NMFS16constraints,Soliman12amb},
modeling the computation of Petri net invariants, siphons and other structural properties as a constraint satisfaction problem
provides efficient implementations using general purpose constraint solvers,
often showing better efficiency than with dedicated algorithms.
This was illustrated here by the use of a constraint logic program to implement our condition on P-invariants and critical siphons
by constraining the search to those sets of places that satisfy the condition,
without having to actually compute the sets of all P-invariants and critical siphons.

Finally, it is also worth noting that beyond verifying the rate-independence property of a CRN
and identifying the output species for which the computed function is rate-independent,
our graphical conditions may also be considered as structural constraints to satisfy for the design of rate-independent CRNs
in synthetic biology~\cite{CAFRM18msb}. They should thus play an important role in CRN design systems in the future.

\subsection*{Acknowledgements}
This work was jointly supported by ANR-MOST \emph{BIOPSY Biochemical Programming System} grant ANR-16-CE18-0029
and ANR-DFG \emph{SYMBIONT Symbolic Methods for Biological Networks} grant ANR-17-CE40-0036.

\bibliographystyle{splncs03}
\bibliography{contraintes}

\end{document}